\documentclass[a4paper,11pt]{article}

\usepackage{a4wide}

\usepackage[english]{babel}

\usepackage{amsfonts}
\usepackage{amsmath}
\usepackage{graphicx}
\usepackage[colorlinks,bookmarks=true]{hyperref}
\usepackage{amsthm}
\usepackage{float}
\usepackage{footnote}
\newtheorem{lemma}{Lemma}[section]
\newtheorem{definition}{Definition}[section]

\newtheorem{proposition}{Proposition}[section]
\newtheorem{remark}{Remark}[section]

\newtheorem{assumption}{Assumption}

\newcommand{\E}{\mathbb{E}}

\title{Liquidity and Impact in Fair Markets
}

\author{Thibault Jaisson\thanks{I would like to thank my PhD supervisors Emmanuel Bacry and Mathieu Rosenbaum for their help in the realization of this work.}\\ CMAP, \'Ecole Polytechnique Paris\\ thibault.jaisson@polytechnique.edu}

\begin{document}

\maketitle

\begin{abstract}
\noindent We develop a theory which applies to any market dynamics that satisfy a fair market assumption on the nullity of the average profit of simple market making strategies. We show that for any such fair market, there exists a martingale fair price which corresponds to the average liquidation value (at the ask or the bid) of an infinitesimal quantity of stock. We show that this fair price is a natural reference price to compute the ex post gain of limit orders. Using only the fair market assumption, we link the spread to the impact of market orders on the fair price. We use our definition of the fair price to build empirical tests of the relevance of this notion whose results are consistent with our theoretical predictions.
\end{abstract}

\noindent
\textbf{Keywords:}
Fair price,
market orders,
limit orders,
bid-ask spread,
liquidity,
price impact,
market efficiency.

\section{Introduction}

One can argue that the main aim of financial markets is to give investors the possibility to balance their portfolio as cheaply and as quickly as possible. To do that, most modern financial exchanges allow their participants to buy or sell assets on an electronic order book via a continuous double auction: Throughout the day, participants can either use limit orders which enable them to buy or sell a given number of shares at a given price but with no guarantee of execution, thus adding liquidity to the order book, or they can use market orders which consume available limit orders thus taking liquidity from the order book. Understanding the nature and dynamics of this liquidity is of great importance for practitioners who wish to improve the efficiency of their strategies. For example, brokers who execute large trades need to split them\footnote{Such large trades are typically much larger than the liquidity available in the order book.} into many transactions and to wait for new liquidity to arrive between their orders. It is also crucial for regulators and exchanges who aim at improving the global efficiency of the market, for example by tuning the tick value\footnote{On electronic markets, assets cannot be exchanged at any price but only at certain discrete prices. The tick value of an asset is the discretization step, see \cite{dayri2012large}.} or transaction costs.\\

\noindent
Several statistical models of order books have been proposed to understand the dynamics of liquidity. Among them, we can cite ``zero intelligence'' models, see \cite{cont2013price}, \cite{farmer2003doubleauction} or \cite{rocsu2009dynamic}, which assume that market and limit orders arrive at Poisson rates. More recently, this topic has been tackled using mean field games, see \cite{lasry2013efficiency} or general Markovian dynamics, see \cite{huang2013simulating}. In \cite{bacry2014estimation}, the arrival of these orders is modeled as Hawkes processes whose kernels are empirically fitted to reproduce the causality between the different events. However, in most of these models, simple market making strategies can yield significant profits on average which we will see is not compatible with empirical computations, especially for small tick assets\footnote{A small tick asset is an asset whose average spread is significantly larger than the tick.}. In this work, we do not give a complete modeling of order book dynamics but rather conditions that reasonable models should satisfy. We show that given the impact of market orders on a fair price that we will define, we can derive the bid-ask spread dynamics.\\

\noindent
The most common measure of liquidity is the spread\footnote{The spread is the difference between the price of the lowest sell limit order (ask price) and the price of the highest buy limit order (bid price) in the order book.}. The links between the spread and other market variables (such as the volatility, the daily volume, the risk aversion of market makers or the transaction costs) have already been the subject of much literature. The classical view is that limit orders are mainly posted by market makers who sell at the ask and buy at the bid and thus ``earning the spread'', while market orders are used by other agents which can be seen as brokers or investors who want to execute a number of shares for various reasons (investment strategies, hedging, regulatory constraints,...).\\

\noindent
In the literature, three main reasons are proposed to explain the size of the spread: The first one is order processing costs, see \cite{huang1997components}. However, in modern electronic financial markets, these costs are negligible compared to the spread\footnote{At least ten times smaller according to \cite{wyart2008relation}.}. The second one is the market makers' risk aversion  (see for example \cite{avellaneda2008high} and \cite{gueant2013dealing} where the average return of market makers is balanced by the inventory risk they take on). The third reason, which we will develop in the next paragraphs, is adverse selection. Throughout this work, we perform empirical measures which seem to show that this reason is sufficient to explain most of the spread behavior.
\\

\noindent
Let us now explain where this adverse selection comes from. When an institutional investor executes a buy metaorder\footnote{A metaorder is a large transaction incrementally executed by an investor, see \cite{farmer2013efficiency}.}, the price on average goes up. This phenomenon, called the market impact of metaorders, has been subject to many empirical and theoretical studies, see \cite{almgren2005direct}, \cite{farmer2013efficiency}, \cite{kyle1985continuous} or  \cite{moro2009market}. The impact of metaorders can be understood in two different ways: It can either be viewed as informational in the sense that investors have a better information on the future of the price and trade because of this information, thus revealing the price, see \cite{farmer2013efficiency} or \cite{kyle1985continuous}. Market impact can also be seen as mechanical in the sense that even uninformed transactions move prices\footnote{See \cite{jaisson2014market} for a more detailed discussion on this topic.} through the imbalance in supply and demand they create.\\

\noindent
Whatever the reason for this impact, during the execution of their metaorders, investors mostly use market orders and therefore, the impact of metaorders is translated onto the order flow into an impact of market orders. This impact has also been the subject of many empirical and theoretical studies, see \cite{biais1995empirical}, \cite{bouchaud2004fluctuations}, \cite{eisler2012price} or \cite{madhavan1997why}. Because of this impact that market orders have on the price, the expected ex post gain of a limit order with respect to the mid price is not equal to the half spread: It is on average equal to the half spread minus the average market impact of a market order on the mid price. This phenomenon is called adverse selection.\\

\noindent
Theoretical studies using adverse selection to explain the bid-ask spread are often rooted to \cite{glosten1985bid}, where the bid and ask prices are set by market makers so that the expected ex post gain of their limit orders with respect to the fundamental value of the stock is equal to zero. In \cite{madhavan1997why}, the authors propose a simple $AR(1)$-type order flow model where the impact on a ``fair price'' is proportional to the surprise in the order flow\footnote{That is the sign of the next order minus its conditional expectation, see below.} so that the fair price is a martingale. In this model, bid and ask prices are set so that the expected ex post gain of a limit order with respect to the fair price is equal to zero. In \cite{wyart2008relation}, Bouchaud et al. apply these ideas to the more realistic FARIMA volume model, see \cite{beran1994statistics}. Assuming that most of price moves are due to volume impact and not exogenous information, they obtain nice relationships between the spread, market impact and the volatility per trade which are consistent with empirical measures. This theory has been extended to large tick assets in \cite{dayri2012large}.\\

\noindent
In most of the models cited above and in many other studies of the behavior of liquidity, spread and impact, there exists an unobservable reference price (called for example fundamental value \cite{madhavan1997why}, efficient price \cite{delattre2013estimating}, \cite{robert2011new} or fair price \cite{gueant2013dealing}). This price often represents some martingale underlying value estimated by market participants based on the available information of the stock based on the available information and is used as a reference price to compute ex post gains. However, it is often not clear what this price really means and why it is compared to this price that market making strategies should not have positive ex post gains.
In this paper, we assume that market dynamics are given (that is a model for limit and market orders) where limit orders can be posted at any price (in practice this corresponds to small tick assets). We will see that in this framework, we in fact do not need to postulate the existence of a reference price. Instead, we only assume that our market dynamics satisfy a fair market assumption which states that \textit{infinitesimal market making strategies are not profitable on average}.\\

\noindent
Our first aim is to show that for any such model, there exists a martingale fair price which generalises that of \cite{madhavan1997why}. This fair price can be thought of as the expected price for buying or selling an infinitesimal quantity of stock using limit orders. We will see that this fair price has nice properties which make it a natural reference price to compute ex post gains. Our second aim is to show that with no other assumptions, we can compute relations between the impact of market orders on this fair price and the liquidity in the order book. We revisit important theoretical results on the spread in light of our fair price. We show that even though on data we cannot compute independently of a model the impact function of market orders on the fair price or even the value of the fair price, we can estimate average ex post gains compared to our fair price. We obtain that our theoretical predictions hold empirically.\\

\noindent
This paper is organized as follows. In Section \ref{LIFM:fairprice}, we specify our no arbitrage hypothesis and show the existence of a model independent fair price. We propose a way to empirically check the existence of our fair price that we apply to real data. In Section \ref{LIFM:spread}, we give a general definition of the market impact of market orders. Using our fair price as a reference price instead of the mid price, we show that we can easily compute the bid and ask prices from the fair price dynamics. Furthermore, we empirically check that adverse selection and tick value explain most the bid-ask spread dynamics. We conclude in Section \ref{LIFM:conclusion}.

\section{Fair market and fair price}
\label{LIFM:fairprice}
In this section, we present our framework and our fair market assumption. We notably show the first of its implications: The existence of a model independent fair price that can be used to compute the ex post gains of limit orders.

\subsection{General framework}

Long term investors rebalance their portfolios using metaorders. At a given time, there is a priori no reason why the flows of buy and sell metaorders should match. Therefore, for the market to be fluid, that is for investors to always find counterparts to execute their trades, there must be intermediaries, called market makers, who provide liquidity at the bid and at the ask in the order book. Since there are many market makers, competition between them should imply that the spread and the liquidity are set so that simple market making strategies are at most marginally ``profitable'' (the profitability of a strategy remaining to be defined). Indeed, if a simple market making strategy were profitable, other market makers would place limit orders just in front of the limit orders of the strategy thus making it less profitable until it is no longer worth doing so, see \cite{wyart2008relation}.\\

\noindent
In what follows, we assume that market makers are risk neutral (in practice, this can be justified by the small time scales of market making strategies) and order processing costs are null (in practice negligible). Therefore, the profitability of a market making strategy is defined as its average profit.\\

\noindent
The concept of ``average'' seen as expectation being undefined outside of a model, we assume that we are given market dynamics. This means we have a probabilistic model of buy and sell market orders and liquidity functions that we define below.

\begin{definition}
We set $v^A_t$ (resp. $v^B_t$) as the volume of buy (resp. sell) market orders at time $t$: If there is a buy (resp. sell) market order of volume $v$ at time $t$ then $v^A_t=v$ (resp. $v^B_t=v$) else $v^A_t=0$ (resp. $v^B_t=0$). The processes $v^A_t$ and $v^B_t$ are assumed to be equal to zero except over a finite number points over any finite time interval.
\end{definition}

\begin{definition}
We set $Liq^{ask}_t$ (resp. $Liq^{bid}_t$) as the ask (resp. bid) cumulated liquidity function at time $t$: For any price $P$, $Liq^{ask}_t(P)$ (resp. $Liq^{bid}_t(P)$) is the sum of the volumes of all ask (resp. bid) limit orders whose prices are smaller (resp. higher) than $P$ at time $t$. These two processes are assumed to be right continuous with left limits.
\end{definition}

\begin{definition}
We set 
\begin{center}
$a_t=\text{argmin}_x\{Liq^{ask}_t(x)>0\}$ (resp. $b_t=\text{argmax}_x\{ Liq^{bid}_t(x)>0\}$)
\end{center}
as the ask (resp. bid) price at time $t$.
\end{definition}

\begin{definition}
\label{defmarket}
The market $M$ is defined as the process: $$M_t=(\sum_{s\leq t}v^A_s,\sum_{s\leq t}v^B_s,Liq^{ask}_t,Liq^{bid}_t)$$ which is right continuous with left limits in time. We write $\mathcal{F}=(\mathcal{F}_t)_t$ for the natural filtration of $M$ and $\mathcal{F}^*$ for the natural filtration of the left continuous (in $t$) version\footnote{The necessity to introduce such a filtration will appear in Section \ref{LIFM:spread}.} of $M$.
\end{definition}

\noindent Note that $\mathcal{F}$ and $\mathcal{F}^*$ are defined so that $\mathcal{F}_t^*\subset \mathcal{F}_t$ and if there is a trade (or any order book event) at time $t$, it is known from $\mathcal{F}_t$ but not from $\mathcal{F}^*_t$.\\

\noindent
We also choose the convention that if a market maker sets a limit order at time $t$ and if a market order arrives at time $t$, then the limit order is executed (provided price priority is respected).\\

\noindent
In this work, we take the tick value equal to zero (that is limit orders can be posted at any price). In practice, this assumption corresponds to small tick assets.

\begin{definition}
An infinitesimal market making strategy is a strategy that uses only very small\footnote{Say of one share, so that this strategy does not impact the rest of the market.} limit orders and that ends with no inventory.
\end{definition}

\noindent
Instead of making an assumption on the nullity of the ex post gain of limit orders compared to an arbitrary price, we directly make an assumption on the profit of market making strategies. This fair market hypothesis on the market dynamics can be understood as a perfect competition hypothesis between market makers:

\begin{assumption}
\label{LIFM:mmhyp}
An infinitesimal market making strategy does not affect the market dynamics and the expected profit of an infinitesimal market making strategy is zero.
\end{assumption}

\noindent
We use this formal assumption on the market dynamics to define a \textit{fair price} which will be the expectation of the price at which an agent can buy or sell an infinitesimal quantity of stocks using limit orders. Our fair price can be viewed as a generalization of the fair price of the model of \cite{madhavan1997why} that we will present in the next paragraph as an illustration of our fair market framework.

\begin{remark}
In \cite{donier2012market}, the same kind of arguments about perfect competition between market makers is used to compute properties of the impact of metaorders.
\end{remark}

\subsection{The example of the MRR model}

In this subsection, we present the Madhavan Richardson Roomans (MRR for short) market model introduced in \cite{madhavan1997why}, which is a fair market model by construction. We compute the expectation of what we call the next bid and ask prices after time $n$. This simple computation illustrates how a fair price appears in any fair market.

\subsubsection{The MRR model}

The MRR model is a simple impact model in which orders are autocorrelated but prices are martingale. Moreover, the ask (resp. bid) price is defined as the expected future fair price if there is a buy (resp. sell) order. Therefore, by construction, the ex post gain of limit orders with respect to the martingale fair price is null and Assumption \ref{LIFM:mmhyp} is satisfied.\\

\noindent
In this discrete time model, the $n^{th}$ market order is either a unit volume buy ($\varepsilon_n=+1$) or sell ($\varepsilon_n=-1$) order and the sign process is an $AR(1)$-type process:
$$\E[\varepsilon_n|\mathcal{F}_{n-1}]=\rho \varepsilon_{n-1}\footnote{Therefore, $\mathbb{P}[\epsilon_n=+1|\mathcal{F}_{n-1}]=(\rho \varepsilon_{n-1}+1)/2$.}$$ with $0<\rho<1$.\\

\noindent
The market impact of an order on the fair price $p$ is proportional to the surprise of the order flow in the sense that:

$$p_{n+1}-p_n=\theta (\epsilon_n-\E[\epsilon_n|\mathcal{F}_{n-1}])+\zeta_n=\theta (\epsilon_n-\rho \epsilon_{n-1})+\zeta_n$$
with $\theta >0$ and $\mathcal{F}_{n-1}$ is the sigma algebra generated by $(\varepsilon_k)_{k<n}$ and $(p_k)_{k\leq n}$ and where $\zeta$ is an independent white noise corresponding to the effect of outside information on the price that we set here to zero to simplify.\\

\noindent
In this model, imposing an ex post gain equal to zero for bid and ask limit orders gives the ask price $$a_n=p_n+\theta (1-\rho \epsilon_{n-1}),$$ the bid price $$b_n=p_n+\theta (-1-\rho \epsilon_{n-1})\footnote{Note than $a_n$ and $b_n$ are in $\mathcal{F}_{n-1}$.}$$ and thus the spread $$\Phi_n=a_n-b_n=2\theta.$$

\subsubsection{Next bid and ask prices}

In the MRR model, bid and ask prices are taken so that the (expected) ex post gain of limit orders with respect to the fair price is equal to zero. Let us now show that conversely, given the dynamics of the ask (or bid) price, we can redefine the fair price. To do this, we need to introduce the following quantities:

\begin{definition}
The next ask (resp. bid) price after time $n$ denoted by $na_n$ (resp. $nb_n$) is the price of the first ask market order after $n$ ($n$ included). More formally:

$$na_n=a_{\inf\{k\geq n;\epsilon_k=+1\}}~ \emph{and} ~nb_n=b_{\inf\{k\geq n;\epsilon_k=-1\}}.$$

\end{definition}

\noindent
Given the dynamics of $(a,b,\epsilon)$, we can derive the fair price as the conditional expectation of the next ask or the next bid price (since they are equal).

\begin{proposition}
\label{LIFM:p10}
We have $\E[na_n|\mathcal{F}_{n-1}]=\E[nb_n|\mathcal{F}_{n-1}]=p_n$.
\end{proposition}

\begin{proof}
Because of the Markov property of the sign process and the invariance of the model with respect to translation of the fair price, the conditional expectation of the next bid price writes as
$$\E[nb_n|\mathcal{F}_{n-1}]=p_n+f^b(\varepsilon_{n-1})$$
for some deterministic function $f^b$. Let us decompose this expectation:
$$\E[nb_n|\mathcal{F}_{n-1}]=\E[\mathbb{I}_{\varepsilon_n=+1}nb_{n+1}|\mathcal{F}_{n-1}]+\E[\mathbb{I}_{\varepsilon_n=-1}b_n|\mathcal{F}_{n-1}]$$
and use that:
\begin{itemize}
\item{$b_n$ is in $\mathcal{F}_{n-1}$ so that $$\E[\mathbb{I}_{\varepsilon_n=-1}b_n|\mathcal{F}_{n-1}]=b_n\mathbb{P}[\varepsilon_n=-1|\mathcal{F}_{n-1}].$$}
\item{By the tower property of conditional expectations
\begin{eqnarray*}
\E[\mathbb{I}_{\varepsilon_n=+1}nb_{n+1}|\mathcal{F}_{n-1}]&=&\E[\mathbb{I}_{\varepsilon_n=+1}\E[nb_{n+1}|\mathcal{F}_{n-1},\varepsilon_n=+1]|\mathcal{F}_{n-1}]\\
&=&\mathbb{P}[\varepsilon_n=+1|\mathcal{F}_{n-1}]\E[nb_{n+1}|\mathcal{F}_{n-1},\varepsilon_n=+1]\\
&=&\mathbb{P}[\varepsilon_n=+1|\mathcal{F}_{n-1}](p_n+\theta (1-\rho \varepsilon_{n-1})+ f^b(+1)). 
\end{eqnarray*}}
\end{itemize}
We thus get that for $\varepsilon_{n-1}=\pm 1$,
\begin{eqnarray*}
\E[nb_n|\mathcal{F}_{n-1}]=p_n+f^b(\varepsilon_{n-1})&=&(p_n+\theta (-1-\rho \epsilon_{n-1}))\times (1-\rho\varepsilon_{n-1})/2\\
&+&(p_n+\theta (1-\rho \varepsilon_{n-1})+ f^b(+1))\times (1+\rho\varepsilon_{n-1})/2.
\end{eqnarray*}
Taking $\varepsilon_{n-1}=+1$, this yields that $f^b(+1)=0$ and then taking $\varepsilon_{n-1}=-1$, this yields that $f^b(-1)=0$ and therefore $\E[nb_n|\mathcal{F}_{n-1}]=p_n$.\\


\noindent In the same way, we can show that $\E[na_n|\mathcal{F}_{n-1}]=p_n$.
\end{proof}

\noindent
Let us now show that Proposition \ref{LIFM:p10} does not come from the particular structure of the MRR model but only from the nullity of the average profit of market making strategies in the MRR model. This implies that the fair price can be generalized to any market model satisfying Assumption \ref{LIFM:mmhyp}.

\subsection{General definition of the fair price}

In a more general way, let us assume that we have known continuous time dynamics for the ask price, the bid price and market orders and that they satisfy Assumption \ref{LIFM:mmhyp}. We begin by defining in continuous time the next bid and ask prices.

\begin{definition}
The next ask (resp. bid) price after time $t$ denoted $na_t$ (resp. $nb_t$) is the price of the first ask market order after (including) $t$. More formally: $$na_t=a_{\inf\{\tau\geq t;v_\tau^A>0\}}$$ and $$nb_t=b_{\inf\{\tau\geq t;v_\tau^B>0\}}.$$
\end{definition}

\noindent We have the following important result.

\begin{proposition}
\label{LIFM:p1}
Under Assumption \ref{LIFM:mmhyp}, we have: $$\E[na_t|\mathcal{F}_{t}^*]=\E[nb_t|\mathcal{F}_{t}^*].$$
\end{proposition}

\begin{proof}

If $\E[na_t|\mathcal{F}^*_{t}]>\E[nb_t|\mathcal{F}^*_{t}]$, maintaining two limit orders just in front of the best ask and the best bid prices from $t$ until they are executed yields a positive profit on average. In the same way, if $\E[na_t|\mathcal{F}^*_{t}]<\E[nb_t|\mathcal{F}^*_{t}]$, this strategy yields a negative profit on average.

\end{proof}

\begin{remark}
This proposition might at first sight seem unnatural because the ask price is always above the bid price but we have shown that it is true in the MRR model and we will see that it is well satisfied on data.
\end{remark}


\noindent We can now define the fair price.

\begin{definition}
The fair price $P_t$ is defined by:
\begin{equation}
\label{LIFM:deffairprice}
P_t=\E[na_t|\mathcal{F}_{t}^*]=\E[nb_t|\mathcal{F}_{t}^*].
\end{equation}
\end{definition}

\noindent We have the first following result about the fair price.

\begin{proposition}
The fair price is a $\mathcal{F}^*$-martingale: $$\E[P_{t+s}|\mathcal{F}_{t}^*]=\E[na_{t+s}|\mathcal{F}_{t}^*]=\E[na_t|\mathcal{F}_{t}^*]=P_t.$$ 
\end{proposition}

\begin{proof}
If $\E[na_{t+s}|\mathcal{F}^*_{t}]>\E[na_t|\mathcal{F}_{t}^*]=\E[nb_t|\mathcal{F}^*_{t}]$, posting a bid order in $t$ just in front of the best bid and an ask order in $t+s$ just in front of the best ask gives a positive profit on average.
\end{proof}

\noindent
It is thus possible to obtain the fair price from any model in which there is an ask price, a bid price and an order flow process which satisfy Assumption \ref{LIFM:mmhyp}. We will see in the next section that it is possible to proceed the other way around and to get the bid and ask prices from any impact model on a fair price (as it is done in the construction of the MRR framework).\\


\noindent
Moreover, since the fair price is a martingale, the response function associated to it defined almost as in \cite{bouchaud2004fluctuations} as $$R(\delta)=\E[P_{t+\delta}-P_t|v^A_t>0]$$ does not depend on $\delta$. Indeed, for $\delta_2>\delta_1$,$$R(\delta_2)-R(\delta_1)=\E[\E[P_{t+\delta_2}-P_{t+\delta_1}|\mathcal{F}^*_{t+\delta_1}]|v^A_t>0]=0.$$ We will later see that this makes the fair price a natural reference price to compute impact functions.

\begin{remark}
In the previous proofs, we have assumed that we could always place ourselves just in front of the ask or bid prices. In fact there is the tick value/priority issue which implies that one cannot place limit orders at any position in the order book and thus Proposition \ref{LIFM:p1} should become: $$|\E[na_t|\mathcal{F}_{t}^*]-\E[nb_t|\mathcal{F}_{t}^*]|\leq 2\alpha$$ where $\alpha$ is the tick value. For many assets, we cannot neglect the tick value but in this work, we assume that it is null.
\end{remark}

\subsection{Verification of Proposition \ref{LIFM:p1} on data}
\label{LIFM:sdata}

\noindent
The aim of this paragraph is to propose a way to empirically check Proposition \ref{LIFM:p1} on real data. By definition of the conditional expectation, Proposition \ref{LIFM:p1}
$$\E[na_t-nb_t|\mathcal{F}^*_{t}]=0$$
implies:

$$\forall E_t\in \mathcal{F}^*_{t}~~\frac{\E[na_t 1_{E_t}]}{\mathbb{P}(E_t)}=\frac{\E[nb_t 1_{E_t}]}{\mathbb{P}(E_t)}.$$

\noindent
By definition of $(\mathcal{F}^*_t)$, we get that for any measurable set $E$ in the image space where $(M_{t-y})_{y>0}$\footnote{Where $M$ is the market process, see Definition \ref{defmarket}.} takes value, denoted $\tilde{\sigma}$, we have

$$\frac{\E[na_t 1_{E}((M_{t-y})_{y>0})]}{\mathbb{P}[(M_{t-y})_{y>0}\in E]}=\frac{\E[nb_t 1_{E}((M_{t-y})_{y>0})]}{\mathbb{P}[(M_{t-y})_{y>0}\in E]}.$$

\noindent Let us now assume that the market is stationary in the sense that for every $E$,
\begin{equation}
\label{eqempiricalver}
\frac{\E[(na_t-nb_t) 1_{E}((M_{t-y})_{y>0})]}{\mathbb{P}[(M_{t-y})_{y>0}\in E]}
\end{equation}
does not depend on $t$.\\

\noindent
To empirically check that Proposition \ref{LIFM:p1} is valid, we need to consider a family of past events $(E^i)_i$ (that are frequent enough) and for every $i$, we use a law of large numbers to show that \eqref{eqempiricalver} is equal to zero.
Here is the list of events that we consider:

\begin{itemize}
\item{$E^0=\tilde{\sigma}$.}
\item{$E^1((M_{t-y})_{y>0})$: The last market order before $t$ is a buy order.}
\item{$E^2((M_{t-y})_{y>0})$: The last market order before $t$ is a sell order.}
\item{$E^3((M_{t-y})_{y>0})$: The last mid price move before $t$ is upward.}
\item{$E^4((M_{t-y})_{y>0})$: The last mid price move before $t$ is downward.}
\end{itemize}

\begin{remark}
We could have taken any event on the past that happens often enough.
\end{remark}

\subsubsection*{Description of our database}

Our database is a list of 25 futures\footnote{Data kindly provided by QuantHouse EUROPE/ASIA, http://www.quanthouse.com.}. Standard notations are used to name the assets. For each of these futures, we consider all the trading days of January 2012. For each of these days, we only take the most liquid maturity. We have a list of all the first limit order book events (market orders, first limit limit orders and first limit cancellations) of the day. We cut off the first and last hours of the trading day.\\

\noindent
We then consider a regular time grid $(t_n)$ of 10 seconds interval\footnote{We could have taken any reasonable time grid. We choose this one to have enough points to approximate the expectation by its empirical average on this time grid. It is useless to take points which are too close together since then they would be too dependent.} and for every $t_n$, we compute the $E^i((M_{t-y})_{y>0})$ and the next bid and ask prices. Averaging over the $N$ points of the month, we can compute the empirical approximations $\Delta^i$ of $\E[(na_t-nb_t) 1_{E^i}((M_{t-y})_{y>0})]/\mathbb{P}(E^i((M_{t-y})_{y>0}))$:
\begin{equation}
\label{LIFM:eqfp}
\Delta^i=\frac{\displaystyle\frac{1}{N}\sum_{k=1}^N 1_{E^i_{t_n}}(na_{t_n}-nb_{t_n})}{\displaystyle\frac{1}{N}\sum_{k=1}^N 1_{E^i_{t_n}}}.
\end{equation}

\noindent We present our results in Table \ref{LIFM:tablefp2}.\\

\begin{savenotes}
\begin{table}[H]
\begin{center}
\begin{tabular}{|c|c|c|c|c|c|c|c|c|}
\hline
Stock & $\Delta^0$ & $\Delta^1$ & $\Delta^2$ & $\Delta^3$ & $\Delta^4$ & Tick & Spread/Tick\\
\hline
xFGBL & 3.31e-3 & 3.33e-3 & 3.29e-3 & 3.33e-3 & 3.29e-3 & 1.e-2 & 1.09 \\
\hline
xFGBS & 3.27e-3 & 3.27e-3 & 3.26e-3 & 3.21e-3 & 3.28e-3 & 5.e-3 & 1.03 \\
\hline
xFGBM & 4.05e-3 & 4.10e-3 & 4.01e-3 & 4.05e-3 & 4.06e-3 & 1.e-2 & 1.04 \\
\hline
xFDAX & 6.26e-2 & 6.30e-2 & 6.20e-2 & 5.76e-2 & 6.75e-2 & 5.e-1 & 1.70 \\
\hline
xFSXE & 6.34e-1 & 6.42e-1 & 6.23e-1 & 6.33e-1 & 6.34e-1 & 1. & 1.03 \\
\hline
xFEURO & 3.42e-5 & 3.41e-5 & 3.43e-5 & 3.40e-5 & 3.44e-5 & 1.e-4 & 1.16 \\
\hline
xFAD & 3.27e-5 & 3.96e-5 & 2.38e-5 & 3.42e-5 & 3.11e-5 & 1.e-4 & 1.26 \\
\hline
xFJY & 7.67e-7 & 5.9e-7 & 2.92e-6 & 8.36e-7 & 7.02e-7 & 1.e-6 & 1.21 \\
\hline
xFBP & 2.74e-5 & 2.11e-5 & 3.15e-5 & 2.94e-5 & 2.53e-5 & 1.e-4 & 1.40 \\
\hline
xFCAD & 2.94e-5 & 2.68e-5 & 3.10e-5 & 3.01e-5 & 2.87e-5 & 1.e-4 & 1.21 \\
\hline
xFCH & 3.25e-5 & 3.14e-5 & 3.04e-5 & 2.52e-5 & 3.83e-5 & 1.e-4 & 1.56 \\
\hline
xFSP & 1.66e-1 & 1.69e-1 & 1.63e-1 & 1.68e-1 & 1.65e-1 & 2.5e-1 & 1.01 \\
\hline
xFND & 1.01e-1 & 9.77e-2 & 1.04e-1 & 1.00e-1 & 1.01e-1 & 2.5e-1 & 1.21 \\
\hline
xFDJ & 4.43e-1 & 4.47e-1 & 4.36e-1 & 4.40e-1 & 4.46e-1 & 1. & 1.39 \\
\hline
xFGOLD & 3.18e-2 & 3.08e-2 & 3.31e-2 & 3.10e-2 & 3.26e-2 & 1.e-1 & 1.93 \\
\hline
xFCOPP & 2.31e-4 & 2.28e-4 & 2.34e-4 & 2.43e-4 & 2.19e-4 & 5.e-4 & 2.13 \\
\hline
\end{tabular}

\end{center}

\caption{Expectations of the difference between the next bid and the next ask for 5 different events during January 2012.}
\label{LIFM:tablefp2}
\end{table}
\end{savenotes}

\noindent
As we have theoretically shown under the fair market assumption, we observe that for all events, the difference is smaller than one tick. Therefore for small tick assets, it is marginal and our prediction is empirically quite satisfied.

\subsection{The fair price as a reference price}

In a rather large class of market dynamics, we have shown that it is always possible to define a martingale ``fair price'' which is the expectation of the price at which it is possible for agents to infinitesimally buy or sell assets using limit orders. 

\begin{definition}
We define the expected ex post gain of an ask (resp. bid) limit order compared to a price process as the expectation of the difference between the price of the limit order and the price process just after the limit order (resp. the difference between the price process after the limit order and the price of the limit order).
\end{definition}

\noindent
In this work, the most important property of the fair price is the following result. It explains why it should be considered as a rigorous reference price to compute ex post gains.

\begin{proposition}
\label{LIFM:epgpro}
Under Assumption \ref{LIFM:mmhyp}, the expected ex post gain of a limit order compared to the fair price is equal to zero.
\end{proposition}

\begin{proof}
Let us assume that at time $t$, the expected ex post gain of an ask limit order is strictly positive. This means that the expectation $\mathcal{E}$ of the fair price just after the execution of this order is smaller than the execution price of the order $A$. Consider the following infinitesimal market making strategy:
\begin{itemize}
\item{Set an ask limit order just in front of this ask limit order.} 
\item{Right after the eventual execution of this ask limit order, set a bid limit order which follows the best bid.}
\end{itemize}
Then, by definition of the fair price (expectation of the price of the next trade at the bid), the expected execution price $B$ of the bid limit order is the value of the fair price right after the execution of the ask limit order $\mathcal{E}$. Therefore, the expected profit of this strategy is $A-B=A-\mathcal{E}>0$.
\end{proof}

\noindent
In the next section, we apply Proposition \ref{LIFM:epgpro} to limit orders in various contexts.

\section{Bid-ask spread, fair price and impact}
\label{LIFM:spread}

\subsection{Necessary and sufficient conditions on the bid and ask prices}

In the previous section, we have given a way to define a fair price and we have shown that the ex post gain of any limit order compared to this price has to be null under the fair market hypothesis. Here, we proceed the other way around and show that given a model of the impact of the flow of market orders on the martingale fair price, we can derive no arbitrage bid and ask prices (as it is done in the construction of the MRR model). To do so, we proceed in the same vein as in \cite{madhavan1997why} and \cite{wyart2008relation} applying Proposition \ref{LIFM:epgpro} to limit orders placed at the best bid and ask prices.

\begin{proposition}
\label{LIFM:p2}
Given dynamics for the fair price and the order flow $(P_t,v^A_t,v^B_t)$, the bid and ask prices at time $t$, $a_t$ and $b_t$ must satisfy:
\begin{subequations}
\begin{equation}
\label{LIFM:ask}
a_t=\E[P_{t^+}|v^A_t>0,\mathcal{F}_{t}^*]\footnote{The notation $\E[P_{t^+}|v^A_t>0,\mathcal{F}_{t}^*]$ means $\underset{\varepsilon\rightarrow 0+}{\lim}\E[P_{t+\varepsilon}|\tau^A_t\in[t,t+\varepsilon],\mathcal{F}^*_t]$ that we assume to always exist, where $\tau^A_t$ is the time of the next ask order after (including) $t$.}
\end{equation}
and
\begin{equation}
\label{LIFM:bid}
b_t=\E[P_{t^+}|v^B_t>0,\mathcal{F}_{t}^*].
\end{equation}
\end{subequations}
\end{proposition}

\begin{proof}
If $a_t>\E[P_t^+|v^A_t>0,\mathcal{F}_{t}^*]$, placing an infinitesimal ask limit order just in front of $a_t$ between $t$ and $\min(t+\varepsilon,\tau^A_t)$ gives a positive average ex post gain for $\varepsilon$ small enough.
\end{proof}

\noindent
In spite of the recurrence of the concept of market impact of market orders in the microstructure literature, there is no consensus on how it should be defined or even on what price it should be computed. In our framework, we can give a natural and model independent definition of the local average market impact of a trade.

\begin{definition}
For any given market dynamics, we define the local average market impact at time $t$ of a bid or ask market order as:
$$AMI^{ask}_t=\E[P_{t^+}|v^A_t>0,\mathcal{F}^*_{t}]-P_t$$
$$AMI^{bid}_t=\E[P_{t^+}|v^B_t>0,\mathcal{F}^*_{t}]-P_t.$$
\end{definition}

\noindent
By calling it local, we want to stress that this value can depend on the market state at time $t$ and is thus difficult to compute empirically.

\begin{remark}
For any $h>0$, we could have taken the definitions:
$$AMI^{ask}_t=\E[P_{t+h}|v^A_t>0,\mathcal{F}^*_{t}]-P_t$$and
$$AMI^{bid}_t=\E[P_{t+h}|v^B_t>0,\mathcal{F}^*_{t}]-P_t.$$

\noindent
Indeed, since $P$ is a martingale, the definition of the impact does not depend on the horizon. This independence is another good reason to consider the impact of market orders on the fair price and not on an arbitrary price as it is often done.

\end{remark}

\noindent The following result is immediate.

\begin{proposition}
\label{LIFM:corspread}
With this definition of market impact, Proposition \ref{LIFM:p2} can be rewritten:
$$a_t=P_t+AMI^{ask}_t$$
and
$$b_t=P_t+AMI^{bid}_t.$$

\end{proposition}

\begin{remark}
This implies the dynamics of $(a_t,b_t,v^A_t,v^B_t)$ can be deduced from the dynamics of $(P_t,v^A_t,v^B_t)$. There is thus an equivalence between these two dynamics under the fair market hypothesis.
\end{remark}

\noindent
Proposition \ref{LIFM:corspread} implies that the spread
$$\Phi_t=a_t-b_t=AMI^{ask}_t-AMI^{bid}_t$$
is completely explained by adverse selection. This is in fact not surprising since Assumption \ref{LIFM:mmhyp} essentially says that adverse selection ``compensates'' the position of liquidity that market makers set in the order book. In the next paragraph, we show that we can test Proposition \ref{LIFM:p2} on data and we will see that it is indeed well satisfied (especially for small tick assets).

\subsection{Verification of Proposition \ref{LIFM:p2} on data}
\label{LIFM:32}

As it is not possible to compute model independently the fair price, neither is it possible to do it for the average market impact of trades $AMI$. However let us show that we can empirically compute adequate ex post gains and thus check Propositions \ref{LIFM:p2}.\\

\noindent
By definition of $P_t$, Equations \eqref{LIFM:ask} and \eqref{LIFM:bid} can be rewritten:
$$a_t=\E[P_{t^+}|v^A_t>0,\mathcal{F}^*_{t}]=\E[na_{t^{+}}|v^A_t>0,\mathcal{F}^*_{t}]\footnote{Which in fact means $\underset{\varepsilon\rightarrow 0+}{\lim}\E[na_{t+\varepsilon}|\tau^A_t\in[t,t+\varepsilon],\mathcal{F}^*_t]$.}$$ and $$b_t=\E[P_{t^+}|v^B_t>0,\mathcal{F}^*_{t}]=\E[nb_{t^{+}}|v^B_t>0,\mathcal{F}^*_{t}].$$

\noindent
Let us now proceed as in Section \ref{LIFM:sdata}. We need the following lemma:
\begin{lemma}
Set
$$a_t^\varepsilon=\E[na_{t+\varepsilon}|\tau^A_t\in[t,t+\varepsilon],\mathcal{F}_t^*]$$
and assume for all $E_t\in \mathcal{F}_t^*$,
$$\E[(a_t^\varepsilon-a_t) 1_{E_t}1_{\tau^A_t\in[t,t+\varepsilon]}]/\mathbb{P}[\tau^A_t\in[t,t+\varepsilon]]\rightarrow 0$$
which is true in any reasonable model since $a^\varepsilon_t\rightarrow a_t$. Then, we have
$$\forall E\in \tilde{\sigma},\text{   }\frac{\E[a_t 1_{E}((M_{t-y})_{y>0})|v^A_t>0]}{\mathbb{P}[(M_{t-y})_{y>0}\in E|v^A_t>0]}=\frac{\E[na_{t^+} 1_{E}((M_{t-y})_{y>0})|v^A_t>0]\footnote{Where the numerator means $\underset{\varepsilon\rightarrow 0+}{\lim}\E[na_{t+\varepsilon}1_{E}((M_{t-y})_{y>0})|\tau^A_t\in[t,t+\varepsilon]]$ and the denominator means $\underset{\varepsilon\rightarrow 0+}{\lim}\mathbb{P}[(M_{t-y})_{y>0}\in E |\tau^A_t\in[t,t+\varepsilon]]$.}}{\mathbb{P}[(M_{t-y})_{y>0}\in E|v^A_t>0]}$$
and
$$\forall E\in \tilde{\sigma},\text{  }\frac{\E[b_t 1_{E}((M_{t-y})_{y>0})|v^B_t>0]}{\mathbb{P}[(M_{t-y})_{y>0}\in E|v^B_t>0]}=\frac{\E[nb_{t^+} 1_{E}((M_{t-y})_{y>0})|v^B_t>0]}{\mathbb{P}[(M_{t-y})_{y>0}\in E|v^B_t>0]}.$$
\end{lemma}

\begin{proof}
By definition of the conditional expectation for all $E_t\in \mathcal{F}_t^*$,
$$\frac{\E[a_t^\varepsilon 1_{E_t}1_{\tau^A_t\in[t,t+\varepsilon]}]/\mathbb{P}[\tau^A_t\in[t,t+\varepsilon]]}{\mathbb{P}[E_t,\tau^A_t\in[t,t+\varepsilon]]/\mathbb{P}[\tau^A_t\in[t,t+\varepsilon]]}
=
\frac{\E[na_{t+\varepsilon} 1_{E_t}1_{\tau^A_t\in[t,t+\varepsilon]}]/\mathbb{P}[\tau^A_t\in[t,t+\varepsilon]]}{\mathbb{P}[ E_t,\tau^A_t\in[t,t+\varepsilon]]/\mathbb{P}[\tau^A_t\in[t,t+\varepsilon]]}.$$
Going to the limit as $\varepsilon$ tends to zero and replacing the event $E_t$ by $(M_{t-y})_{y>0}\in E$ as in Section \ref{LIFM:sdata} ends the proof.
\end{proof}

\noindent Let us now assume that $$\frac{\E[(a_t-na_{t^+}) 1_{E}((M_{t-y})_{y>0})|v^A_t>0]}{\mathbb{P}[(M_{t-y})_{y>0}\in E|v^A_t>0]}$$ and $$\frac{\E[(b_t-nb_{t^+}) 1_{E}((M_{t-y})_{y>0})|v^B_t>0]}{\mathbb{P}[(M_{t-y})_{y>0}\in E|v^B_t>0]}$$ do not depend on $t$.\\

\noindent
Considering the same $E^i$ as before and summing over all the $N^{ask}$ and $N^{bid}$ ask and bid market order times $(t^{ask}_n)$ and $(t^{bid}_n)$ of the assets of our database during January 2012, we need to show that
\begin{subequations}
\begin{equation}
\label{LIFM:eqbidask}
\Delta A^i=\frac{\displaystyle\frac{1}{N^{ask}} \sum_{k=1}^{N^{ask}} 1_{E^i_{t^{ask}_n}}(a_{t^{ask}_n}-na_{(t^{ask}_n)^+})}{\displaystyle\frac{1}{N^{ask}} \sum_{k=1}^{N^{ask}} 1_{E^i_{t^{ask}_n}}}\simeq  0
\end{equation}
and
\begin{equation}
\label{LIFM:eqbidbid}
\Delta B^i=\frac{\displaystyle\frac{1}{N^{bid}}\sum_{k=1}^{N^{bid}} 1_{E^i_{t^{bid}_n}}(b_{t^{bid}_n}-nb_{(t^{bid}_n)^+})}{\displaystyle \frac{1}{N^{bid}}\sum_{k=1}^{N^{bid}} 1_{E^i_{t^{bid}_n}}}\simeq  0.
\end{equation}
\end{subequations}
In Table \ref{LIFM:tablebidask}, we list the empirical $\Delta A^i$ and $\Delta B^i$ for our list of futures after the events $(E^i)$ defined by Equations \eqref{LIFM:eqbidask} and \eqref{LIFM:eqbidbid}.\\
\begin{table}[H]
\begin{center}
\begin{tabular}{|c|c|c|c|c|c|c|}
\hline
Stock & $\Delta A^0$ & $\Delta A^1$ & $\Delta B^0$ & $\Delta B^1$ & Tick & NOrders \\
\hline
xFGBL & -1.00e-6 & 1.91e-4 & -1.09e-6 & -2.08e-4 & 1.e-2 & 8.00e5 \\
\hline
xFGBS & -1.16e-6 & 6.00e-5 & 4.72e-7 & -7.95e-5 & 5.e-3 & 1.98e5 \\
\hline
xFGBM & -1.08e-6 & 1.62e-4 & 7.41e-7 & -1.84e-4 & 1.e-2 & 3.11e5 \\
\hline
xFDAX & 1.35e-3 & 8.10e-3 & 1.48e-3 & -5.77e-3 & 5.e-1 & 1.23e6 \\
\hline
xFSXE & 3.74e-4 & 2.05e-2 & 4.36e-4 & -2.66e-2 & 1. & 8.27e5 \\
\hline
xFEURO & 2.44e-8 & 1.15e-6 & 2.52e-8 & -1.02e-6 & 1.e-4 & 1.63e6 \\
\hline
xFAD & 1.64e-7 & 7.66e-7 & 1.59e-7 & -3.74e-7 & 1.e-4 & 6.37e5 \\
\hline
xFJY & 4.01e-8 & 4.41e-8 & 6.39e-10 & 4.14e-10 & 1.e-6 & 1.27e4 \\
\hline
xFBP & 1.49e-7 & 2.93e-7 & 1.29e-7 & -2.55e-7 & 1.e-4 & 4.66e5 \\
\hline
xFCAD & 1.43e-7 & 5.90e-7 & 1.45e-7 & -1.68e-7 & 1.e-4 & 3.98e5 \\
\hline
xFCH & 2.47e-7 & -9.47e-7 & 1.66e-7 & 1.34e-6 & 1.e-4 & 1.90e5 \\
\hline
xFSP & 2.73e-5 & 2.56e-3 & 2.63e-5 & -2.64e-3 & 2.5e-1 & 2.98e6 \\
\hline
xFND & 3.73e-4 & 2.51e-3 & 3.61e-4 & -1.32e-3 & 2.5e-1 & 1.12e6 \\
\hline
xFDJ & 7.97e-4 & 2.53e-3 & 9.24e-4 & 1.22e-3 & 1. & 7.33e5 \\
\hline
xFGOLD & 1.88e-3 & 6.89e-3 & 6.11e-4 & 1.53e-4 & 1.e-1 & 9.02e5 \\
\hline
xFCOPP & 2.08e-6 & 6.10e-6 & 2.15e-6 & -1.79e-6 & 5.e-4 & 3.66e5 \\
\hline
\end{tabular}
\end{center}
\caption{Average next ask (resp. bid) after an ask (resp.bid) market order compared to the ask (resp. bid) price just before the order after some events.}
\label{LIFM:tablebidask}
\end{table}

\noindent As expected, we see that the $\Delta A^i$ and $\Delta B^i$ are always significantly smaller than a tick.\\

\noindent
The nullity of average ex post gains of limit orders placed at the best prices for small tick assets shows that simple market making strategies can only be marginally profitable on average. Therefore, we have empirically shown that, as postulated in \cite{madhavan1997why} and \cite{wyart2008relation}, adverse selection and the tick value explain the bid and ask prices and thus the spread. The risk aversion of market makers and their order processing costs seem to have negligible influences on the spread.

\subsection{Response functions, time horizon and market making}

In \cite{wyart2008relation}, similar conditions on the profit of market making strategies are used except that the ex post gain of limit orders is computed compared to the mid price a time $\delta t$ after the order: $E[a_t-m_{t+\delta t}|v_t^A>0]$ for ask limit orders and $E[m_{t+\delta t}-b_t|v_t^B>0]$ for bid limit orders.\\

\noindent However, the mid price is not a martingale in general. In particular, the empirical response function is not constant but significantly (compared to the tick value) increases over time, see \cite{bouchaud2004fluctuations}. Therefore, the ex post gain of market makers should depend on the horizon of their strategies. In particular, it is stated that fast market making should be more profitable than slow market making.\\

\noindent This is inconsistent with our framework where the average profit of a market making strategy does not depend on the ``horizon'' of the market maker.\\

\noindent In order to empirically show that the average profit of market making strategies does not seem to depend much on the horizon of market makers, we plot the ask (resp. bid) response function of the fair price defined as:
$$RA^P(\delta)=\E[P_{t+\delta}-P_t|v^A_t>0]$$ and $$RB^P(\delta)=\E[P_{t+\delta}-P_t|v^B_t>0]$$
which also write
$$RA^P(\delta)=\E[na_{t+\delta}-P_t|v^A_t>0]$$ and $$RB^P(\delta)=\E[nb_{t+\delta}-P_t|v^B_t>0]$$
by definition of the fair price.\\

\noindent
In our framework, this function should not depend on $\delta$ since $P$ is a martingale. Let us notice that since we cannot empirically compute model independently the fair price before the orders, we cannot measure $RA^P$ and $RB^P$. However, we can compute the empirical variations of the ask (resp. bid) response function on the fair price defined as the empirical average over the ask and bid order times of January 2012:
$$RA^P(\delta)-RA^P(0^+)=\frac{1}{N^{ask}}\sum_{i=1}^{N^{ask}} (na_{t^{ask}_n+\delta}-na_{(t^{ask}_n)^+})$$ and $$RB^P(\delta)-RB^P(0^+)=\frac{1}{N^{bid}}\sum_{i=1}^{N^{bid}} (nb_{t^{bid}_n+\delta}-nb_{(t^{bid}_n)^+})$$
for three liquid assets, see Figures \ref{LIFM:resp1}, \ref{LIFM:resp2} and \ref{LIFM:resp3}.

\begin{figure}[H]
\centering
\includegraphics[width=90mm,height=56mm]{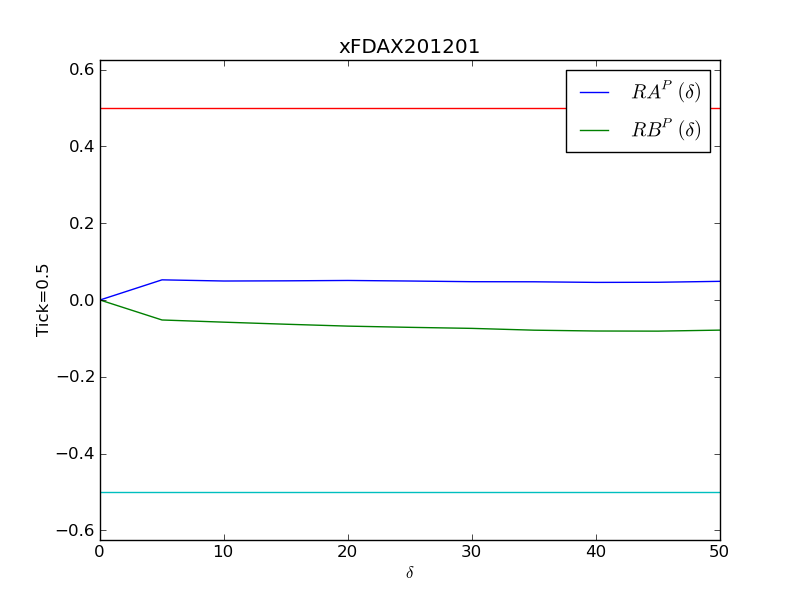}
\caption{Ask (in blue) and bid (in green) response functions for the DAX future. The two horizontal lines correspond to plus (in red) and minus (in cyan) the tick value.}
\label{LIFM:resp1}
\end{figure}

\begin{figure}[H]
\centering
\includegraphics[width=90mm,height=56mm]{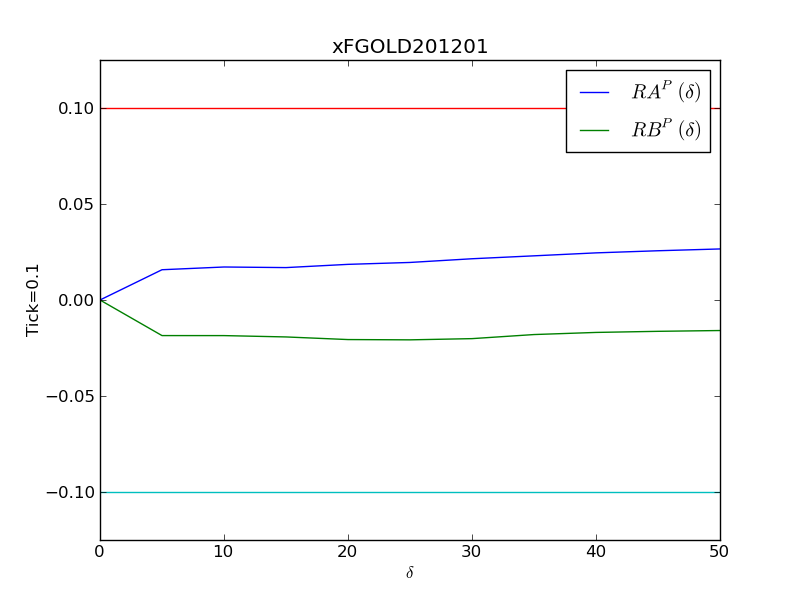}
\caption{Ask (in blue) and bid (in green) response functions for the GOLD future. The two horizontal lines correspond to plus (in red) and minus (in cyan) the tick value.}
\label{LIFM:resp2}
\end{figure}

\begin{figure}[H]
\centering
\includegraphics[width=90mm,height=56mm]{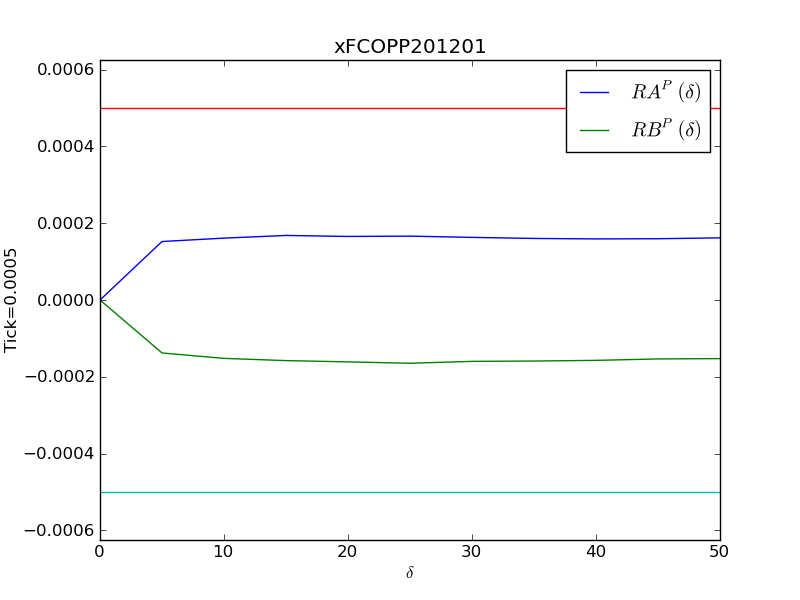}
\caption{Ask (in blue) and bid (in green) response functions for the COPP future. The two horizontal lines correspond to plus (in red) and minus (in cyan) the tick value.}
\label{LIFM:resp3}
\end{figure}

\noindent We observe on Figures \ref{LIFM:resp1}, \ref{LIFM:resp2} and \ref{LIFM:resp3} that in accordance with our predictions, the variations of the response functions computed on the fair price  are small compared to the tick value.

\begin{remark}
Notice that there is a slight increment of the ask response function at very small time scales (below the second) which implies that fast market making might be marginally more profitable than slow market making. However, this increment is always small compared to the tick value.
\end{remark}

\subsection{On tick value, market making and liquidity}

We have shown evidence that for small tick assets, the distance of some liquidity from the best price was perfectly compensated by the impact of market orders which consumed this liquidity. This gave ex post gains close to zero for market makers. For large tick assets, it is no longer the case. Indeed, market makers cannot any more perfectly compete by putting liquidity as close to the mid price as they might wish therefore, the competition is in the speed at which they take the priority.\\

\noindent
This seems to make the markets less efficient. On the other hand, a large tick market might attract market makers and thus liquidity. Moreover, a large tick simplifies the information flow that practitioners get. The optimal tick value, should be computed with these arguments in mind.
See \cite{dayri2012large} for an example of computation of an optimal tick value.


\section{Conclusion}
\label{LIFM:conclusion}


We have studied continuous time market dynamics which only satisfy a fair market hypothesis. Theoretically, this hypothesis can be justified by arguments of perfect competition between market makers. In this general framework, we have given precise and model independent definitions of concepts such as fair price, market impact and liquidity that are often present in the literature on order book dynamics. Taking our fair price as a reference price, we have derived a relationship between, market impact and liquidity which is supported by real data analysis.\\

\noindent
From a practical point of view, the strategies of market participants such as brokers depend a lot on liquidity and impact. Such agents should thus calibrate their strategies using models which accurately reproduce stylized relationships between these quantities. In particular, in such models, the ex post gains of limit orders compared to a martingale reference price should be equal to zero.\\

\noindent
Our framework implies that market making strategies are not profitable. Obviously, this is not rigorously true because it is the competition between market makers that makes our hypothesis reasonable. The first reason for the coherence between our framework and the existence of market makers is the tick value which can lead to profits for market makers. The second reason is that we have only considered very simple strategies and small statistical arbitrages may remain for more complex strategies. Finally some market participants are encouraged to provide liquidity in order to obtain advantages for other market activities.


\bibliographystyle{abbrv}
\bibliography{bibliofp}

\end{document}